\documentclass[aps,pra,reprint,twocolumn,amsmath,amssymb,10pt,nofootinbib,superscriptaddress]{revtex4-2}

\usepackage{amsmath}
\usepackage{amssymb}
\usepackage{amsthm}
\usepackage[hidelinks]{hyperref}
\usepackage{graphicx}
\usepackage[caption=false]{subfig}
\usepackage{physics}
\usepackage{bbold}
\usepackage{xcolor}
  \definecolor{mydarkblue}{HTML}{0072b2}
  \definecolor{myorange}{HTML}{e69f00}
  \definecolor{bluishgreen}{HTML}{009e73}
  \definecolor{reddishpurple}{HTML}{cc79a7}
  \definecolor{mid-darkblue-orange}{HTML}{738959}
\newcommand{\ii}{\mathrm{i}}
\newcommand{\M}{\mathcal{M}}
\usepackage{tikz}
  \usetikzlibrary{decorations.pathmorphing, decorations.pathreplacing, decorations.shapes, calc}
\usepackage{tikz-3dplot}
  \tdplotsetmaincoords{70}{120}

\newtheorem{theorem}{Theorem}
\newtheorem{lemma}{Lemma}

\begin{document}

\title{Uncountably many inequivalent maximally entangled measurements for two qutrits}
\author{Elna Svegborn}
\thanks{These authors share the first authorship.}
\affiliation{Physics Department and NanoLund, Lund University, Box 118, 22100 Lund, Sweden}
\author{Jef Pauwels}
\thanks{These authors share the first authorship.}
\affiliation{Department of Applied Physics University of Geneva, 1211 Geneva, Switzerland}
\affiliation{Constructor University, Bremen, Germany}
\author{Nicolas Gisin}
\affiliation{Department of Applied Physics University of Geneva, 1211 Geneva, Switzerland}
\affiliation{Constructor University, Bremen, Germany}
\author{Alejandro Pozas-Kerstjens}
\affiliation{Department of Applied Physics University of Geneva, 1211 Geneva, Switzerland}

\begin{abstract}
    Every two-qubit measurement basis composed of maximally entangled eigenstates can be transformed into the Bell basis via local unitary operations.
    For higher dimensions, in contrast, there exist inequivalent bases composed of maximally entangled eigenstates.
    Here, we provide a single-parameter family of two-qutrit maximally entangled measurement bases, and demonstrate that none of the bases are equivalent to each other under local unitaries.
    These bases are constructed from the continuous family of symmetric informationally complete sets of states in dimension three.
    By studying the local unitary bases that generate the family of two-qutrit maximally entangled measurement bases, we construct the first examples of wild error bases in the smallest dimension where these can exist.
    Finally, we discuss how distinct measurements in the family lead to differences in performance in several scenarios relevant in quantum information.
\end{abstract}

\maketitle

A maximally entangled measurement is a quantum measurement whose eigenstates are all maximally entangled.
Such measurements are a central resource in quantum information science and play a key role in paradigmatic protocols such as quantum teleportation \cite{bennett1993}, dense coding \cite{bennett1992}, and entanglement swapping \cite{zukowski1993}.
The seminal example is the Bell state measurement, which projects a pair of qubits onto the orthonormal Bell basis.
This is also the unique two-qubit maximally entangled measurement, up to local unitary transformations \cite{Vollbrecht2000}.

However, when the particles have more than two internal levels, there exists in general multiple inequivalent classes of maximally entangled measurements \cite{Vollbrecht2000,Werner2001,popp2024}.
This implies that not all such measurement bases can be transformed to the generalized Bell basis of the same dimension via local operations \cite{baumgartner2006,popp2024}.
The existence of inequivalent, maximally entangled measurement bases leads to three natural questions: (i) how can we identify whether two maximally entangled measurements belong to distinct equivalence classes? (ii) which property makes two classes distinct from one another? and (iii) are different classes useful for different tasks?
Addressing these questions is both interesting for understanding the nature of maximally entangled measurements, and useful to know when a particular class of these measurements can offer advantages over another class.

In addition, identifying distinct classes of maximally entangled bases is equivalent to finding distinct orthonormal bases of unitary operators \cite{Vollbrecht2000}.
These can provide insights into the geometry of Hilbert space \cite{baumgartner2006}, in a similar way as mutually unbiased bases or symmetric informationally complete (SIC) sets of states do.
Such unitary bases are also broadly relevant in quantum science: they are relevant for quantum error correction codes \cite{knill1996,klappenecker2003,Klappenecker2005},and they provide a way to generate nonlocality in quantum networks \cite{networkReview}.

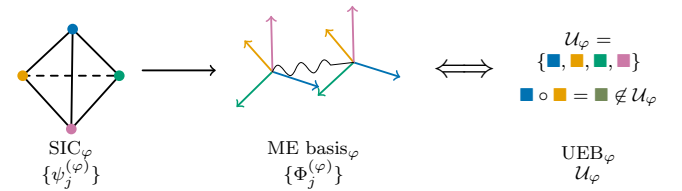
\begin{figure}
    \centering
    \resizebox{\linewidth}{!}{%
    \begin{tikzpicture}[
        line cap=round,
        line join=round,
        every node/.style={font=\footnotesize}
    ]
        \coordinate (A) at (-3.8,0.62);
        \coordinate (B) at (-4.55,-0.08);
        \coordinate (C) at (-3.1,-0.08);
        \coordinate (D) at (-3.82,-0.88);

        \draw[thick] (B) -- (A) -- (C);
        \draw[thick,dashed] (B) -- (C);
        \draw[thick] (A) -- (D);
        \draw[thick] (B) -- (D);
        \draw[thick] (C) -- (D);

        \fill[mydarkblue] (A) circle[radius=.08cm];
        \fill[myorange] (B) circle[radius=.08cm];
        \fill[bluishgreen] (C) circle[radius=.08cm];
        \fill[reddishpurple] (D) circle[radius=.08cm];
        \node[align=center] at (-3.82,-1.42) {SIC$_\varphi$\\[-1pt]$\{\psi_j^{(\varphi)}\}$};

        \draw[->,thick] (-2.75,0.00) -- node[above=2pt] {} (-1.65,0.00);

        \coordinate (O1) at (-0.80,-0.02);
        \coordinate (A1) at (-0.10,-0.25);
        \coordinate (B1) at (-1.22,0.45);
        \coordinate (C1) at (-1.34,-0.55);
        \coordinate (D1) at (-0.82,0.82);
        \coordinate (O2) at (0.42,0.12);
        \coordinate (A2) at (1.12,-0.10);
        \coordinate (B2) at (0.00,0.62);
        \coordinate (C2) at (-0.12,-0.38);
        \coordinate (D2) at (0.40,0.97);
        \draw[->,thick,mydarkblue] (O1) -- (A1);
        \draw[->,thick,myorange] (O1) -- (B1);
        \draw[->,thick,bluishgreen] (O1) -- (C1);
        \draw[->,thick,reddishpurple] (O1) -- (D1);
        \draw[->,thick,mydarkblue] (O2) -- (A2);
        \draw[->,thick,myorange] (O2) -- (B2);
        \draw[->,thick,bluishgreen] (O2) -- (C2);
        \draw[->,thick,reddishpurple] (O2) -- (D2);
        \draw[decorate,decoration={coil,aspect=0}] (O1) -- (O2);
        \node[align=center] at (-0.18,-1.42) {ME basis$_\varphi$\\[-1pt]$\{\Phi_j^{(\varphi)}\}$};

        \node at (2.08,0.00) {\Large$\Longleftrightarrow$};

        \node at (3.95,0.02) {
            \begin{minipage}{2.25cm}
                \centering
                $\mathcal{U}_{\varphi}=\{{\color{mydarkblue}{\blacksquare}},{\color{myorange}{\blacksquare}},{\color{bluishgreen}{\blacksquare}},{\color{reddishpurple}{\blacksquare}}\}$\\[5pt]
                ${\color{mydarkblue}{\blacksquare}}\circ{\color{myorange}{\blacksquare}}={\color{mid-darkblue-orange}{\blacksquare}}\not\in\mathcal{U}_{\varphi}$
            \end{minipage}
            };
        \node[align=center] at (3.95,-1.42) {UEB$_\varphi$\\[-1pt]$\mathcal{U}_\varphi$};
    \end{tikzpicture}%
    }
    \caption{
        Summary of the work.
        The starting point are SIC sets of states.
        From any of them, Ref.~\cite{Czartowski2021} gives a way for constructing families of bipartite quantum measurements.
        For local dimensions 2 and 3, the corresponding family contains maximally entangled measurements, and since in dimension 3 there exists a continuous family of SICs, the construction leads to a continuous family of maximally entangled measurements.
        The sets of unitaries that generate the eigenstates form, in general, wild error bases, that are the smallest examples known so far.
        By characterizing the equivalence of these bases under unitary transformations, we determine which of the measurements in the family under local unitary operations.
    }
    \label{fig:summary}
\end{figure}

In this work we focus on the most elementary, nontrivial setting for studying distinct classes of maximally entangled measurements, namely two-qutrit measurements.
We give a single-parameter family of two-qutrit measurement bases composed of maximally entangled states.
This family is obtained from qutrit SIC sets of states using the construction of Ref.~\cite{Czartowski2021}, and yields the generalized two-qutrit Bell state measurement as a special case.
Within the family, we identify the equivalence classes under local unitary operations, i.e., the unique different measurements that it encodes.
We identify that the family contains a continuously infinite amount of inequivalent measurement bases, which stands in stark contrast with the two-qubit case.

Furthermore, by studying the structure of the family, we find that only the generalized Bell basis and its equivalents are generated by nice unitary error bases, which are projectively closed under multiplication \cite{Werner2001, popp2024,klappenecker2003,Klappenecker2005}.
In contrast, all other members of the family are generated from so-called \textit{wild}, or \textit{wicked}, error bases \cite{klappenecker2003}.
Whether an error basis is nice or wild is relevant, for instance, when correcting for accumulating errors in quantum computations and in quantum repeater chains.
To this end, we present the first examples of wild error bases in dimension three, the previous known results being in dimension four or higher \cite{beckman2001,Klappenecker2005,musto2016}.

Thereafter, we analyze the complexity of the measurements in the family by studying the entanglement cost required for its remote implementation \cite{Pauwels2025}.
Despite each member of the family satisfying the necessary conditions for ideal remote implementation \cite{popescu1994}, we find that none, beyond those equivalent to the Bell state measurement, can be implemented ideally between distant parties, regardless of the amount of pre-shared entanglement.
Nonetheless, we identify the members of the family which admit a non-ideal, efficient localization.

Lastly, we demonstrate that the different maximally entangled measurements also present operational differences.
We illustrate this in the context of quantum computing (in particular, in magic state injection protocols), and in that of quantum communication via quantum repeaters.

\paragraph*{From SICs to maximally entangled measurements.---}
A SIC in dimension $d$ is a set of $d^2$ pure states, $\{\ket{\psi_i}\}_{i=1}^{d^2}$, with the property that the overlap between any pair of such states is identical, i.e., $\left|\braket{\psi_i}{\psi_j}\right|^2=\frac{1+d\delta_{ij}}{1+d}$ for all $i,j$ \cite{Renes2004}.
From any SIC we can construct a two-particle eigenbasis, $\{|\Phi_i^{(\alpha)}\rangle\}_{i=1}^{d^2}$, where the $i$-th eigenvector is constructed from the $i$-th state in the set as follows~\cite{Czartowski2021}:
\begin{equation}
  \ket*{\Phi_i^{(\alpha)}}=\sqrt{1+\frac{1}{d}}\ket{\psi_i,\psi_i^*}-e^{\ii\alpha}f_\pm(\alpha)\ket{\Omega^+}.
  \label{eq:family}
\end{equation}
Here, $\ket{\psi_i^*}$ denotes the complex conjugate of $\ket{\psi_i}$, $\ket{\Omega^+} = \frac{1}{\sqrt{d}}\sum_{j=0}^{d-1}\ket{j,j}$ is the maximally entangled unit vector, and $f_\pm(\alpha)=\frac{1}{d}\left[\sqrt{(d+1)\cos^2\alpha}\pm\sqrt{(d+1)\cos^2\alpha-d}\right]$ is a function in the real parameter $\alpha$.
By construction, each of the eigenvectors $\ket*{\Phi_i^{(\alpha)}}$ has the same degree of entanglement, meaning that the measurement basis is isoentangled \cite{santo2024}.

The eigenvectors \eqref{eq:family} are maximally entangled if and only if $\cos\alpha=\pm\frac{\sqrt{d+1}}{2}$, see Appendix \ref{app:SICmeas} for details.
This implies that the eigenstates in Eq.~\eqref{eq:family} can only generate a maximally entangled basis in dimensions $d=2,\,3$, since the magnitude of $\cos \alpha$ can at most reach unity.
For two-qubit systems, it is well-known that all maximally entangled bases are equivalent to the standard Bell basis under local unitaries \cite{Vollbrecht2000}.
Thus, we focus on the construction obtained with $\alpha = 0$ in $d = 3$.

In dimension three there exists a continuous family of SICs \cite{Renes2004}.
This family can be obtained from the action of the Weyl-Heisenberg group on a fiducial state characterized by a single parameter $\varphi$, i.e., 
\begin{equation}
    \ket*{\psi_{j}^{(\varphi)}} = X^{j_1}\cdot Z^{j_2} \ket*{\phi^{(\varphi)}},
    \label{eq:family_SIC}
\end{equation}
where $j \equiv j_1 j_2 \in \{0,1,2\}^2$, and $\ket*{\phi^{(\varphi)}} = \frac{1}{\sqrt{2}}(\ket{1}-e^{\ii \varphi}\ket{2})$ is a valid fiducial state for any $\varphi\in \left[0,\pi/3\right]$.
The group generators can be chosen as the so-called shift and clock matrices, namely $X=\sum_{i=0}^2\ketbra{i\oplus1}{i}$, and $Z=\sum_{i=0}^2\omega^i\ketbra{i}{i}$ with eigenvectors$\omega=e^{ 2\pi \ii /3}$.
Importantly, any two such SICs, $\{|\psi_j^{(\varphi_1)}\rangle\}_j$ and $\{|\psi_j^{(\varphi_2)}\rangle\}_j$, are equivalent if and only if $\varphi_2=\frac{2\pi}{9}\pm\varphi_1$ \cite{Zhu2010}.

Inserting the expression for $\ket*{\psi_{j}^{(\varphi)}}$ in Eq.~\eqref{eq:family_SIC} into Eq.~\eqref{eq:family}, and taking $\alpha = 0$, we obtain a continuous set of maximally entangled measurements parameterized by $\varphi$.
For fixed $\varphi$, the corresponding eigenstates are
\begin{equation}
    \ket*{\Phi_j^{(\varphi)}}=\frac{2}{\sqrt{3}}\ket*{\psi_j^{(\varphi)},\psi_j^{(\varphi)}{}^*}-\ket{\Omega^+},\qquad j=1,\ldots,9.
    \label{eq:familyM}
\end{equation}
We give the explicit form of the measurement eigenstates in Appendix~\ref{app:M}.

Moreover, any maximally entangled basis can be generated from an orthonormal single-site unitary basis acting on a maximally entangled unit state \cite{Vollbrecht2000}.
Taking this unit state to be $\ket*{\Phi^{(\varphi)}_1}$ in Eq.~\eqref{eq:familyM}, each eigenstate of the measurement for phase $\varphi$ can be expressed as $ \ket*{\Phi_j^{(\varphi)}} = (U_j^\varphi\otimes\mathbb{1})\ket*{\Phi_1^{(\varphi)}}$ where the unitary matrices $U^\varphi_j$ belong to the set
\begin{equation}
    \mathcal{U}_\varphi =\left\{\, Z^{a},\; Z^{a}\cdot X_{+}^{(\varphi)},\; Z^{a}\cdot X_{-}^{(\varphi)}\ \middle|\ a\in\{0,1,2\}\right\},
    \label{eq:familyU}
\end{equation}
with $Z$ being the standard clock operator and
\begin{equation}
    X_{+}^{(\varphi)}=\begin{pmatrix}
        0 & 1 & 0 \\
        0 & 0 & 1 \\
        e^{3\ii\varphi} & 0 & 0
    \end{pmatrix},\quad
    X_{-}^{(\varphi)}=\begin{pmatrix}
        0 & 0 & 1 \\
        1 & 0 & 0 \\
        0 & e^{-3\ii\varphi} & 0
    \end{pmatrix}.
\end{equation}
These unitaries satisfy the commutation relations
\begin{equation}
    Z\cdot X^{(\varphi)}_+=\omega^{-1}X^{(\varphi)}_+\cdot Z,\qquad
    Z\cdot X^{(\varphi)}_-=\omega X^{(\varphi)}_-\cdot Z,
\end{equation}
meaning that $(Z,X_+^{(\varphi)})$ and $(Z,X_-^{(\varphi)})$ behave as Weyl pairs with opposite commutator phases, up to the central phases $(X_+^{(\varphi)})^3=e^{3\ii\varphi}\mathbb{1}$ and $(X_-^{(\varphi)})^3=e^{-3\ii\varphi}\mathbb{1}$.

\paragraph*{Unitary error bases and equivalent measurements.---}
Out of all the measurement bases in the family, it is important to determine which ones correspond to the same physical measurement (i.e., equivalent modulo local unitary transformations).
If the two-qutrit measurements were built out of non-maximally entangled eigenstates (i.e., if $f_\pm(\alpha)^2<1$), it is easy to see that equivalence of the SICs generating the bases is in one-to-one correspondence to equivalence of the measurements generated (see proof in Appendix~\ref{app:equivSIC}).
However, the case of interest, namely when the bases are maximally entangled, must be treated separately.
The relevance of $\mathcal{U}_\varphi$ is not only that it generates the maximally entangled bases, but also that it can be used as the ground for addressing equivalence.
More precisely, we consider
\begin{equation}
    \ket{\Phi_i}=(U_i\otimes\mathbb{1})\ket{\Omega^+},
    \qquad
    \ket{\Psi_i}=(V_i\otimes\mathbb{1})\ket{\Omega^+},
\end{equation}
for some orthonormal set of unitaries $\{U_i\}_i$ and $\{V_i\}_i$.
These two maximally entangled bases are locally unitarily equivalent if and only if the two unitary bases $\{U_i\}$ and $\{V_i\}$ are equivalent in the standard left-right sense, namely if there exist unitaries $A$, $B$, phases $e^{\ii\theta_i}$, and a permutation $\pi$ such that
\begin{equation}
    U_i=e^{\ii\theta_i}A\cdot V_{\pi(i)}\cdot B \qquad\text{for all }i.
    \label{eq:equiv}
\end{equation}
This is is directly a consequence of the vectorization identity $(C\otimes D)\ket{\Omega^+}=(C\cdot D^T\otimes\mathbb{1})\ket{\Omega^+}$, see Appendix~\ref{app:equiv}.
Hence, it is enough to classify the unitary bases $\mathcal U_\varphi$ up to this equivalence relation.

The complete characterization of equivalent and inequivalent bases within the family is captured in the theorem below
\begin{theorem}\label{thm:equivM}
    For real parameters $\varphi_1$, $\varphi_2$, the two unitary error bases $\mathcal U_{\varphi_1}$ and $\mathcal U_{\varphi_2}$ (and, by extension, the corresponding maximally entangled qutrit bases) are equivalent if and only if
    \begin{equation*}
        \varphi_1\pm\varphi_2=\frac{2\pi}{9}k\quad\text{for some } k\in\mathbb Z.
    \end{equation*}
\end{theorem}

\begin{proof}
In order to determine the equivalence of two bases, $\mathcal U_{\varphi_1}$ and $\mathcal U_{\varphi_2}$, we use the following object
\begin{equation}
    I(\mathcal U) = \sum_{i,j,k,l=1}^9 \left|\Tr\!\left(U_i^\dagger U_jU_k^\dagger U_l\right)\right|^4.
    \label{eq:invariant}
\end{equation}
This object is invariant under local operations of the form of Eq.~\eqref{eq:equiv}.
In Appendix~\ref{app:equiv} we compute it exactly for this family, obtaining
\[
    I(\mathcal U_\varphi)=51273+7776\cos(9\varphi).
\]
Therefore, $\mathcal U_{\varphi_1}$ and $\mathcal U_{\varphi_2}$ being equivalent implies $\cos(9\varphi_1)=\cos(9\varphi_2)$, or, alternatively, $\varphi_1\pm\varphi_2=2\pi k/9$.
Conversely, also in Appendix~\ref{app:equiv} we give the unitaries implementing the transformations $\varphi\mapsto\varphi+2\pi/9$ and $\varphi\mapsto2\pi/9-\varphi$ and these generate all the identifications.
\end{proof}

Note that this condition, $\varphi_1\pm\varphi_2= 2\pi k/9$, is the same as the condition for the generating SICs to be equivalent \cite{Zhu2010}.
Thus, in analogy with the case of SICs, the measurement for any $\varphi\in\left[0,\frac{\pi}{3}\right]$ is only equivalent to a finite number of other measurements.
Therefore, there exists a continuously infinite number of maximally entangled inequivalent measurements for two qutrits.

\paragraph*{Wild unitary error bases for $d=3$.---}
A central consequence of Theorem \ref{thm:equivM} is that the complete characterization of equivalent unitary error bases $\mathcal U_\varphi$ allows us to give the first (in fact, a continuous infinity of them) known examples of wild error bases for $U(3)$.

A nice error basis is, up to the equivalence in Eq.~\eqref{eq:equiv}, a basis arising from a projective unitary representation of a group of order $d^2$.
An example of nice error bases is the standard representation of the Weyl-Heisenberg group, that generates the generalized Bell state measurement.
A wild, or wicked, error basis is one that is not equivalent to any nice error basis \cite{klappenecker2003,Klappenecker2005}.

It turns out that all bases of the form of Eq.~\eqref{eq:familyU} not equivalent to the generalized Bell basis are wild.
Indeed, if $e^{9\ii\varphi}=1$, Theorem~\ref{thm:equivM} makes $\mathcal U_\varphi$ equivalent to $\mathcal U_0$, and $\mathcal U_0=\{Z^a\cdot X^b:a,b\in\mathbb Z_3\}$ (i.e., the standard representation of the Weyl-Heisenberg group).
Conversely, any three-dimensional nice error basis is equivalent to the qutrit Weyl-Heisenberg basis \cite{Werner2001,klappenecker2003}, so if $\mathcal U_\varphi$ were nice it would be equivalent to $\mathcal U_0$, and Theorem~\ref{thm:equivM} then gives $e^{9\ii\varphi}=1$.
Thus, every $\varphi$ with $e^{9\ii\varphi}\neq1$ gives a wild error basis.

Previous works had identified wild error bases for $U(4)$ \cite{beckman2001,Klappenecker2005,musto2016}.
The $\mathcal{U}_{\varphi\not=\frac{2\pi}{9}k}$, which are bases for $d=3$, constitute (to the best of our knowledge) the first known examples of wild error bases for $U(3)$.
Note that this dimension is minimal, since in dimension $d=2$ every unitary error basis is equivalent to the Pauli basis \cite{Werner2001}, which is a nice error basis.

By considering $\mathcal{U}_\varphi$ as a generating set and explicitly computing the associated projective group order (i.e., the size of the resulting set closed under multiplication, modulo global phases) for all values of $0 <\varphi<\frac{\pi}{9}$ that can be written as fractions of $\pi$ with denominator at most 100, we find that the basis that generates the smallest group (after the Weyl-Heisenberg group for $\varphi=0$) is that corresponding to $\varphi=\frac{\pi}{9}$.
The corresponding projective group order is 36.
Whether this order can be reduced via unitary equivalences (recall, Eq.~\eqref{eq:equiv}), and whether there exist other wild error bases that generate groups of smaller order, remain open questions.

\paragraph*{Measurement complexity and localizability.---}
Recently, a new framework has been developed to characterize the complexity of multipartite quantum measurements \cite{Pauwels2025}.
Motivated by studies on the tension between global quantum operations and special relativity \cite{beckman2001,cirac2001,dur2001,vaidman2003,groisman2002,Groisman2003}, this framework quantifies the complexity of a measurement in terms of the nonlocal resources required to perform it on spatially separated subsystems using non-adaptive local quantum operations, classical joint post-processing, and shared entanglement.
While the statistics of any joint measurement can be simulated (or \textit{localized}) in this way with unlimited entanglement \cite{groisman2002,Groisman2003,vaidman2003,clark2010}, only few examples are known that require little entanglement \cite{Pauwels2025}.
In this context it is important to distinguish between reproducing the measurement statistics alone and reproducing the full measurement process, including the post-measurement states.
The latter, known as \textit{ideal localization}, is possible only if the eigenstates of the measurement do not encode any local information \cite{popescu1994}.
This is precisely the case for measurements composed of maximally entangled eigenstates.
In the special case of Bell state measurement, it is further known that the full measurement process can be implemented locally using only two Bell pairs \cite{popescu1994}.

Thus, a natural question is for which values of $\varphi$ the two-qutrit maximally entangled measurements can be localized ideally.
Surprisingly, a direct consequence of Ref.~\cite{akibue2026} is that ideal localization is possible only for those measurements that are equivalent to the generalized Bell state measurement.
This stems from the fact that ideal localization additionally requires that the single-site unitaries generating the measurement eigenstates form a nice error basis \cite{akibue2026}.
The same condition is necessary to localize measurements non-ideally using a single state $\ket{\Omega^+}$ as the source of shared entanglement \cite{akibue2026}.
Thus, localizing any measurement in the family, with the exception of (local-unitary equivalents of) the Bell state measurement, also requires at least two copies of $\ket{\Omega^+}$.

In particular, this means that no measurement in the family but those equivalent to the Bell state measurement belong to the Clifford group.
This is a consequence of the fact that the Clifford group (see its definition in Appendix~\ref{App:Clifford}) is a subset of the set of measurements that can be localized with one single $\ket{\Omega^+}$ \cite{Pauwels2025}.
This is relevant in the generation of nonlocality in networks that distribute maximally entangled states, where it is known that Clifford operations cannot produce nonlocal correlations \cite{Gatto_Lamas_2023}.

Then, how much entanglement is needed to localize a given measurement in the family?
In general, this question is difficult to address \cite{Pauwels2025}.
However, the structure of the family enables simplifications.
By combining Eqs.~\eqref{eq:family} and \eqref{eq:family_SIC}, it is easy to see that the eigenstates of the measurement are generated by the action of the abelian group $\langle Z \otimes Z^*, X \otimes X\rangle$ on the fiducial state $\ket{\Phi_\varphi}=\frac{2}{\sqrt{3}}\ket*{\phi^{(\varphi)}}\otimes\ket*{\phi^{(\varphi)}}-\ket{\Omega^+}$.
This results in the fact that all measurements in the family exhibit so-called \textit{tetrahedral symmetry}, and thus its position in the Clifford hierarchy can easily be determined by studying that of a single diagonal matrix \cite{pauwels2025b}.
In Appendix~\ref{App:localization} we show that the level in the Clifford hierarchy can be easily determined for phases of the type $\varphi=\frac{2\pi k}{3^l}$ for integers $k$ and $l$.
Interestingly, the measurement that generates the wild error basis of the smallest projective group order, $\varphi=\frac{\pi}{9}$, is not of this form and thus does not belong to a finite level of the Clifford hierarchy.

We illustrate the phases of the measurements that appear at low levels in Table~\ref{tab:clifford}.
According to this, the measurement associated with $\varphi = \frac{2\pi}{27}$ is the simplest example of non-Clifford maximally entangled measurement.

\begin{table}[ht!]
    \begin{tabular}{c|cccccccccccccc}
        $\varphi$ & 0 & $\frac{2 \pi}{243}$ & $\frac{4 \pi}{243}$ & $\frac{2 \pi}{81}$ & $\frac{8 \pi}{243}$ & $\frac{10 \pi}{243}$ & $\frac{4 \pi}{81}$ & $\frac{14 \pi}{243}$ & $\frac{16 \pi}{243}$ & $\frac{2 \pi}{27}$ &  $\frac{20 \pi}{243}$ & $\frac{22 \pi}{243}$  &$\frac{8 \pi}{81}$ & $\frac{26 \pi}{243}$  \\[0.2em] \hline
        $k$ & 2 & 9 & 9 & 7 & 9 &  9 & 7 & 9 & 9 & 5 & 9 & 9 & 7 & 9
    \end{tabular}
    \caption{Phases $\varphi$ associated with the maximally entangled measurements in Eq.~\eqref{eq:familyM} that appear at low Clifford level $k$.}
    \label{tab:clifford}
\end{table}

\paragraph*{Operational implications.---} 
We now address the final question, namely whether the different unitary error bases $\mathcal U_\varphi$ lead to operational differences in quantum information tasks.

First, $\mathcal{U}_\varphi$ is a wild error basis whenever $\varphi\not=\frac{2\pi}{9}k$.
In the context of quantum computing, this implies that concatenation of errors in the space of the basis can enlarge the set of necessary correction unitaries \cite{klappenecker2003,Klappenecker2005}.
In fact, since every two compositions of the set with itself generates powers of $X^{(\varphi)}_+\cdot X^{(\varphi)}_-=\text{diag}(1,e^{-3\ii\varphi},e^{3\ii\varphi})$ the set of corrections contains infinitely many distinct unitaries whenever $\varphi/\pi$ is irrational.
This also implies that the choice of unitary error basis becomes important in the context of long-distance entanglement distribution.
If, in a chain of quantum repeaters, the intermediate nodes perform maximally entangled measurements whose generating unitaries correspond to nice error bases, the set of corrections that the end parties need to perform does not increase with the length of the chain.
If they use maximally entangled measurements generated by wild error bases instead, the number of necessary corrections may grow up to indefinitely with increasing chain length.

Second, we note that our family of measurements enables a form of \textit{non-stabilizerness injection} using only stabilizer states.
Non-stabilizerness (or \textit{magic}) is necessary for universal quantum computation \cite{gottesman1998}, and the amount of it in a quantum circuit is connected to the complexity of simulating that circuit classically \cite{Braviy2016}.
In qubit-based quantum computing, non-stabilizerness is inserted into a circuit by performing non-Clifford operations on a qubit and using teleportation \cite{Gottesman_1999}, or by performing distillation protocols \cite{Braviy2005}.
Note that in the teleportation case, the access to single-qubit non-Clifford operations is necessary to inject non-stabilizerness, since all two-qubit maximally entangled measurements are equivalent to the Bell state measurement.
For qutrits, in contrast, our family of measurements allows injection without explicitly necessitating single-qutrit non-Clifford operations.
Thus, the maximally entangled measurements described in this work point to potential, concrete advantages of quantum computing with higher-dimensional systems over that based on qubits.
Note, however, that it is currently not known whether injecting non-stabilizerness via non-Bell maximally entangled measurements is more efficient than performing Bell state measurements complemented with non-Clifford local rotations.
An interesting question for further research is whether different measurements in the family inject, e.g., different amounts of non-stabilizerness.
We expect this to be the case, given that the (in the case of qubits) the localization cost is directly connected to the $T$-gate count \cite{Speelman2016}.

\paragraph*{Discussion.---}
We have identified that there exists an uncountable infinity of two-qutrit maximally entangled measurements that are physically distinct.
This is in stark contrast with the case of two qubits, where all such measurements can be mapped via local unitaries to the Bell state measurement.

The maximally entangled measurements are constructed from SICs, and thus the existence of this continuous family of measurements is directly connected to the existence of a continuous family of SICs in $d=3$.
Whether similar families of maximally entangled measurements exist for higher dimensions is an interesting open question.
Answering it will require, however, novel methods for constructing measurements from SICs, since the construction in Ref.~\cite{Czartowski2021} only produces maximally entangled measurements for $d<4$.

By studying the family of measurements we have identified an associated continuously infinite number of wild error bases for $U(3)$, which is the smallest possible dimension in which wild error bases can exist.
These are the first wild error bases known in this dimension.
Wild error bases are of interest in quantum error correction, and the present construction gives a concrete connection between SICs, maximally entangled measurements, and unitary error bases.

The differences between the bases also lead to different performances in operational tasks.
We identify this behavior in two different quantum information protocols: magic-state injection in quantum computing, and long-distance entanglement distribution via quantum repeaters.
For long-distance entanglement distribution via quantum repeaters, wild bases can require correction sets that grow with the chain length, whereas nice error bases do not.
In quantum computing, however, the measurements in the family seem to allow for magic-state injection protocols that are not available in the qubit-based framework.
It remains an interesting question whether one can build more practical protocols for, e.g., multipartite randomness generation or secret sharing, where using these maximally entangled measurements is more beneficial than using Bell state measurements.

A concrete promising area where these measurements may find interesting applications is that of nonlocality in networks, where non-Clifford resources are necessary to go beyond stabilizer-state protocols: networks composed of stabilizer states (in particular, maximally entangled states) and Clifford measurements (such as Bell state measurements) cannot generate nonlocal correlations \cite{Gatto_Lamas_2023}.
At the same time, practical protocols favor measurements whose localization complexity is not too high.
The family identified here provides maximally entangled measurements that are non-Clifford, while some members still appear at low levels of the Clifford hierarchy.
This makes them promising resources for networks that distribute maximally entangled states.
However, it must be noted that in many networks, when using exclusively maximally entangled states and any type of maximally entangled measurements (thus including those in the family considered) the resulting distributions admit a local model, see Appendix~\ref{app:trianglelocal}.

More broadly, our results emphasize how little is understood about the geometry and entanglement of measurements, as opposed to that of states.
Whereas the structure of bipartite maximally entangled states is completely settled by local unitary equivalence, the space of maximally entangled measurements already displays a continuum of inequivalent classes in the smallest nontrivial dimension.
That this landscape is governed by SICs, objects that sit at the heart of some of the deepest open problems in quantum geometry \cite{Renes2004,Zhu2010}, highlight the richness of this structure, of which the family identified here is likely only a glimpse.

\paragraph*{Acknowledgments.---}
We thank Ingemar Bengtsson for helpful discussions.
This work is supported by the Swiss National Science Foundation (grant number 224561 and NCCR-SwissMAP), and the Swedish Research Council under Contract No. 2023-03498.

\bibliography{references.bib}

\appendix
\onecolumngrid

\section{Bases of maximally entangled states}
\label{app:SICmeas}
In this appendix we derive the conditions under which the construction of Ref.~\cite{Czartowski2021} produces bases of maximally entangled states.
Using the defining property of SICs, $\left|\braket{\psi_i}{\psi_j}\right|^2=\frac{1+d\delta_{ij}}{1+d}$, and that $\frac{1}{\sqrt{d}}\sum_{j=1}^{d}\ket{j,j}=\frac{1}{d\sqrt{d}}\sum_{j=1}^{d^2}\ket{\psi_j,\psi_j^*}$ for any SIC, it is possible to see that the single-party reduced states of the eigenstates of the measurement basis, $\{\ket{\Phi_i}\}$, are
\begin{equation}
  \rho_i^A=\left[1-f_\pm(\alpha)^2\right]\ket{\psi_i}\bra{\psi_i}+f_\pm(\alpha)^2\frac{\mathbb{1}}{d},\qquad\rho_i^B=\left[1-f_\pm(\alpha)^2\right]\ket{\psi^*_i}\bra{\psi^*_i}+f_\pm(\alpha)^2\frac{\mathbb{1}}{d}.
\end{equation}
From here, the purity of the states is easy to compute:
\begin{equation}
  \Tr\left[\left(\rho_i^A\right)^2\right]=\Tr\left[\left(\rho_i^B\right)^2\right]=1-\left(1-\frac{1}{d}\right)\left[2-f_\pm(\alpha)^2\right]f_\pm(\alpha)^2.
\end{equation}

Let us denote this quantity by $\mathcal{P}$.
Note that the purity is independent of the particular SIC used to construct the basis.
In the case of pure bipartite states, $\mathcal{P}$ is a measure of entanglement and achieves its minimum value, $1/d$, for maximally entangled states.

In order to find the smallest value of $\mathcal{P}$, we look at the zeros of $\dv{\mathcal{P}}{\alpha}$.
\begin{equation}
  \begin{aligned}
    \dv{\mathcal{P}}{\alpha}=&-\left(1-\frac{1}{d}\right)\left\{-2 f_\pm(\alpha) f_\pm(\alpha)^2 + \left[2 - f_\pm(\alpha)^2\right] 2 f_\pm(\alpha)\right\} \dv{f_\pm(\alpha)}{\alpha}\\
    =&-4\left(1-\frac{1}{d}\right)f_\pm(\alpha)\left[1-f_\pm(\alpha)^2\right]\dv{f_\pm(\alpha)}{\alpha}.
  \end{aligned}
\end{equation}
This function is zero at several points: wherever $f_\pm(\alpha)=0$ (giving $\mathcal{P}=1$, i.e., product states), wherever $f_\pm(\alpha)^2=1$ (giving $\mathcal{P}=1/d$, i.e., maximally entangled states), and at the points where $\dv{f_\pm(\alpha)}{\alpha}=0$.

Using the definition of $f_\pm(\alpha)$, the solution to $f_\pm(\alpha)^2=1$ is $\cos\alpha=\pm\frac{1}{2}\sqrt{d+1}$ (respectively for $f_\pm$).
Note that $\frac{1}{2}\sqrt{d+1}$ is in $[-1,1]$ (so that $\cos\alpha$ is well defined) for $d\leq 3$ only.
Thus, this construction can only generate maximally entangled bases for $d=2$ and $d=3$.

For completeness we need to check the extreme points of $f_\pm(\alpha)$.
The maximum of $f_+(\alpha)$ occurs at $\alpha=0$, and its minimum occurs at $\alpha=\pi-\arccos\sqrt{\frac{d}{d+1}}\equiv\pi-\alpha_c$.
Analogously, the maximum of $f_-(\alpha)$ occurs at $\alpha=\alpha_c$, and its minimum occurs at $\alpha=\pi$.
In these points, the purity is
\begin{align*}
  \mathcal{P}_+(0)=\mathcal{P}_-(\pi) &= 1-\frac{2}{d}-\frac{1}{d^2}+\frac{11}{d^3}-\frac{8}{d^5}-4\frac{\sqrt{d+1}}{d^2}\left(1-\frac{2}{d}-\frac{1}{d^2}+\frac{2}{d^3}\right),\\
  \mathcal{P}_+(\pi-\alpha_c)=\mathcal{P}_-(\alpha_c)&=1-\frac{2}{d}+\frac{3}{d^2}-\frac{1}{d^3}.
\end{align*}
None of these expressions is equal to $1/d$.
In summary, the SIC construction of isoentangled bases in Ref.~\cite{Czartowski2021} only produces bases of maximally entangled states for $d=2,3$, and it does so when the parameter $\alpha$ satisfies $\cos\alpha=\pm\frac{1}{2}\sqrt{d+1}$.

\section{Explicit form of the measurements}
\label{app:M}
In this appendix we give the explicit expression of the eigenstates of the family of measurements.
We will do so by giving the unitary matrix that realizes the change of basis from the computational basis to the corresponding measurement basis, i.e., $\mathcal{M}_\varphi=\sum_{j=1}^9\ketbra*{\Phi_j^{(\varphi)}}{j}$.
This matrix contains the measurement eigenstates as columns.
We call this matrix the \textit{measurement unitary}.

Inserting Eq.~\eqref{eq:family_SIC} into Eq.~\eqref{eq:family} and setting $\alpha=0$ we find that the corresponding measurement unitary is
\begin{equation}
  \M_\varphi\,{=}\,\frac{1}{\sqrt{3}}\begin{pmatrix}
    0 & 0 & 0 & 1 & 1 & 1 & 0 & 0 & 0 \\
    z^* & z^*\omega^2 & z^*\omega & 0 & 0 & 0 & 0 & 0 & 0 \\
    0 & 0 & 0 & 0 & 0 & 0 & z & z\omega & z\omega^2 \\
    z & z\omega & z\omega^2 & 0 & 0 & 0 & 0 & 0 & 0 \\
    0 & 0 & 0 & 0 & 0 & 0 & 1 & 1 & 1 \\
    0 & 0 & 0 & z^* & z^*\omega^2 & z^*\omega & 0 & 0 & 0 \\
    0 & 0 & 0 & 0 & 0 & 0 & z^* & z^*\omega^2 & z^*\omega \\
    0 & 0 & 0 & z & z\omega & z\omega^2 & 0 & 0 & 0 \\
    1 & 1 & 1 & 0 & 0 & 0 & 0 & 0 & 0
  \end{pmatrix},
  \label{eq:familyStates}
\end{equation}
where $z=e^{\ii\varphi}$.

Applying the single-qutrit unitary $W_1=e^{-\ii\varphi}\ketbra{0}{1} + e^{\ii\varphi}\ketbra{1}{0} + \ketbra{2}{2}$ in the first party one arrives at a simpler notation where the eigenstates have the same support as the Bell states:
\begin{equation}
  (W_1\otimes\mathbb{1})\cdot\M_\varphi\,{=}\,\frac{1}{\sqrt{3}}\begin{pmatrix}
    1 & \omega & \omega^2 & 0 & 0 & 0 & 0 & 0 & 0 \\
    0 & 0 & 0 & 0 & 0 & 0 & z^* & z^* & z^* \\
    0 & 0 & 0 & z^{-2} & \omega^2 z^{-2} & \omega z^{-2} & 0 & 0 & 0 \\
    0 & 0 & 0 & z & z & z & 0 & 0 & 0 \\
    1 & \omega^2 & \omega & 0 & 0 & 0 & 0 & 0 & 0 \\
    0 & 0 & 0 & 0 & 0 & 0 & z^2 & \omega z^2 & \omega^2 z^2 \\
    0 & 0 & 0 & 0 & 0 & 0 & z^* & z^*\omega^2 & z^*\omega \\
    0 & 0 & 0 & z & z\omega & z\omega^2 & 0 & 0 & 0 \\
    1 & 1 & 1 & 0 & 0 & 0 & 0 & 0 & 0
  \end{pmatrix}.
  \label{eq:familyStates2}
\end{equation}
Note that, in fact, the three first states turn into the three Bell states $\left\{Z^a\cdot\frac{1}{\sqrt{3}}\left(\ket{00}+\ket{11}+\ket{22}\right)\right\}_{a=0,1,2}$.
In this case, the single-qutrit unitary error basis that generates the measurement is still $\mathcal{U}_\varphi$ in Eq.~\eqref{eq:familyU}.

In fact, it is possible to apply a further local rotation that gives a form of the measurement basis with six Bell states.
In Eq.~\eqref{eq:familyStates2}, applying any tensor-product unitary $W_2\otimes W_2^*$ does not modify the first three eigenstates.
Choosing $W_2=e^{-\ii\varphi}\ketbra{0}{1}+e^{-2\ii\varphi}\ketbra{1}{2}+\ketbra{2}{0}$ we obtain
\begin{equation}
  \left[(W_2\cdot W_1)\otimes W_2^*\right]\cdot\M_\varphi\,{=}\,\frac{1}{\sqrt{3}}\begin{pmatrix}
    1 & \omega^2 & \omega & 0 & 0 & 0 & 0 & 0 & 0 \\
    0 & 0 & 0 & 0 & 0 & 0 & z^3 & \omega z^3 & \omega^2 z^3 \\
    0 & 0 & 0 & 1 & 1 & 1 & 0 & 0 & 0 \\
    0 & 0 & 0 & 1 & \omega & \omega^2 & 0 & 0 & 0 \\
    1 & 1 & 1 & 0 & 0 & 0 & 0 & 0 & 0 \\
    0 & 0 & 0 & 0 & 0 & 0 & z^{-3} & \omega^2 z^{-3} & \omega z^{-3} \\
    0 & 0 & 0 & 0 & 0 & 0 & 1 & 1 & 1 \\
    0 & 0 & 0 & 1 & \omega^2 & \omega & 0 & 0 & 0 \\
    1 & \omega & \omega^2 & 0 & 0 & 0 & 0 & 0 & 0
  \end{pmatrix}.
  \label{eq:familyStates3}
\end{equation}
Interestingly, in this case the set of generating unitaries is not exactly of the form of Eq.~\eqref{eq:familyU}, but of a unitary-equivalent of it, namely $\{\mathbb{1}, Z, Z^2, X, X\cdot Z, X\cdot Z^2, T_\varphi\cdot X^2, T_\varphi\cdot X^2\cdot Z, T_\varphi\cdot X^2\cdot Z^2\}$, with $T_\varphi=\ketbra{0}+e^{3\ii\varphi}\ketbra{1}+e^{-3\ii\varphi}\ketbra{2}$.

\section{Full characterization of equivalent non-maximally entangled measurements}
\label{app:equivSIC}
In this appendix we analyze the relation between the equivalence of SICs and the equivalence of the measurement bases they generate via Eq.~\eqref{eq:family}.
Let us begin assuming that we have two SICs, $\{\ket{\psi_i}\}$ and $\{\ket{\phi_i}\}$, that are equivalent, i.e., for which there exist a unitary $U$, a permutation $\pi$, and phases $\theta_i$ such that $\ket{\psi_i}=e^{\ii\theta_i}U\ket{\phi_{\pi(i)}}$.
For these SICs, the corresponding bases, $\{\ket{\Phi_i}\}$ and $\{\ket{\Psi_i}\}$, satisfy
\begin{equation}
  \begin{aligned}
    \ket{\Psi_i}&=c\ket{\psi_i,\psi_i^*}-f(\alpha)e^{\ii\alpha}\ket{\Omega^+}\\
    &=\sqrt{1+\frac{1}{d}}\,U\otimes U^*\ket{\phi_{\pi(i)},\phi_{\pi(i)}^*}-f(\alpha)e^{\ii\alpha}\ket{\Omega^+}\\
    &=\sqrt{1+\frac{1}{d}}\,U\otimes U^*\ket{\phi_{\pi(i)},\phi_{\pi(i)}^*}-f(\alpha)e^{\ii\alpha}U\otimes U^*\ket{\Omega^+}\\
    &=U\otimes U^*\left[\sqrt{1+\frac{1}{d}}\ket{\phi_{\pi(i)},\phi_{\pi(i)}^*}-f(\alpha)e^{\ii\alpha}\ket{\Omega^+}\right]\\
    &=U\otimes U^*\ket{\Phi_{\pi(i)}},
  \end{aligned}
\end{equation}
where, in the third line, we have used $A\otimes B\ket{\Omega^+}=A\cdot B^T\otimes\mathbb{1}\ket{\Omega^+}$ to replace $U\cdot U^\dagger\otimes\mathbb{1}\ket{\Omega^+}$ by $U\otimes(U^\dagger)^T\ket{\Omega^+}=U\otimes U^*\ket{\Omega^+}$.

Now, we study the converse.
Let us assume that two bases $\{\ket{\Phi_i}\}$ and $\{\ket{\Psi_i}\}$ are equivalent, i.e., that there exist local unitaries $A$ and $B$, a permutation $\pi$, and phases $\theta_i$, such that $\ket{\Psi_i}=e^{\ii\theta_i}A\otimes B\ket{\Phi_{\pi(i)}}$.
Then, we have that
\begin{equation}
  \begin{aligned}
    \ket{\Psi_i}&=\sqrt{1+\frac{1}{d}}\ket{\psi_i,\psi_i^*}-f(\alpha)e^{\ii\alpha}\ket{\Omega^+}\\
    &=e^{\ii\theta_i}A\otimes B\left[\sqrt{1+\frac{1}{d}}\ket{\phi_{\pi(i)},\phi_{\pi(i)}^*}-f(\alpha)e^{\ii\alpha}\ket{\Omega^+}\right]\\
    &=e^{\ii\theta_i}\sqrt{1+\frac{1}{d}}\, A\otimes B\ket{\phi_{\pi(i)},\phi_{\pi(i)}^*}-f(\alpha)e^{\ii\theta_i}e^{\ii\alpha}A\otimes B\ket{\Omega^+}\\
    &=e^{\ii\theta_i}\sqrt{1+\frac{1}{d}}\, A\otimes B\ket{\phi_{\pi(i)},\phi_{\pi(i)}^*}-f(\alpha)e^{\ii\theta_i}e^{\ii\alpha}\mathbb{1}\otimes B\cdot A^T\ket{\Omega^+}.
  \end{aligned}
\end{equation}
Taking the partial trace over the second subsystem, we find that the marginal states satisfy
\begin{equation}
    \rho_i^A=\left[1-f(\alpha)^2\right]\ket{\psi_i}\bra{\psi_i}+f(\alpha)^2\frac{\mathbb{1}}{d}=\left[1-f(\alpha)^2\right]A\ket{\phi_{\pi(i)}}\bra{\phi_{\pi(i)}}A^\dagger+f(\alpha)^2\frac{\mathbb{1}}{d}.
    \label{eq:marginalNME}
\end{equation}
Note that, for $f(\alpha)^2\not=1$, both expressions are rank-1 perturbations of the maximally mixed state.
Therefore, Eq.~\eqref{eq:marginalNME} implies that $\ket{\psi_i}=e^{\ii\tau_i}A\ket{\phi_{\pi(i)}}$ for some phases $\tau_i$.

These two statements imply that, for $f(\alpha)^2\not=1$, the equivalence of SICs is a necessary and sufficient condition for the equivalence of the corresponding measurement bases.
Importantly, bases composed of maximally entangled states have $f(\alpha)^2=1$ (the single-party marginals are maximally mixed), so the reasoning after Eq.~\eqref{eq:marginalNME} does not apply, and thus, in principle, non-equivalent SICs could give rise to equivalent measurements.

\section{Equivalence between maximally entangled measurements}
\label{app:equiv}

When the bases are composed of maximally entangled states, the marginal-state argument of Appendix~\ref{app:equivSIC} no longer proves necessity.
Instead, we use the standard correspondence between maximally entangled bases and unitary error bases.

We begin by stating a known lemma:
\begin{lemma}[\cite{Werner2001}]\label{lem:meb-ueb-equivalence}
    Let
    \[
        \ket{\Phi_i}=(U_i\otimes\mathbb{1})\ket{\Omega^+},
        \qquad
        \ket{\Psi_i}=(V_i\otimes\mathbb{1})\ket{\Omega^+},
        \qquad
        \ket{\Omega^+}=\frac{1}{\sqrt 3}\sum_{j=0}^2\ket{j,j}.
    \]
    Then the maximally entangled bases $\{\ket{\Phi_i}\}$ and $\{\ket{\Psi_i}\}$ are locally unitarily equivalent if and only if the associated unitary error bases are equivalent, i.e., if and only if there exist unitaries $A,B$, phases $e^{\ii\theta_i}$, and a permutation $\pi$ such that
    \[
        U_i=e^{\ii\theta_i}A\cdot V_{\pi(i)}\cdot B
        \qquad\text{for every }i.
    \]
\end{lemma}

\begin{proof}
    If $\ket{\Phi_i}=e^{\ii\theta_i}(A\otimes B)\ket{\Psi_{\pi(i)}}$, then
    \[
        (U_i\otimes\mathbb{1})\ket{\Omega^+}
        =
        e^{\ii\theta_i}(A\cdot V_{\pi(i)}\cdot B^T\otimes\mathbb{1})\ket{\Omega^+},
    \]
    where we used $(C\otimes D)\ket{\Omega^+}=(C\cdot D^T\otimes\mathbb{1})\ket{\Omega^+}$.
    The vectorization map is injective, so $U_i=e^{\ii\theta_i}A\cdot V_{\pi(i)}\cdot B^T$.
    Since $B^T$ is unitary, this is unitary-error-basis equivalence.
    The converse is the same argument in reverse.
\end{proof}

Now we can proceed to prove the result.
\renewcommand{\thetheorem}{1}
\begin{theorem}\label{thm:app-equiv}
    Two measurement bases of the form of Eq.~\eqref{eq:familyM} for phases $\varphi_1$ and $\varphi_2$ are equivalent under local unitary operations if and only if $\varphi_1\pm\varphi_2=\frac{2\pi k}{9}$ for some $k\in\mathbb Z$.
\end{theorem}

\begin{proof}
For any unitary error basis $\mathcal U=\{U_i\}_{i=1}^{9}$, define
\[
    I(\mathcal U)
    :=
    \sum_{i,j,k,l=1}^{9}
    \left|
        \Tr\!\left(U_i^\dagger U_j U_k^\dagger U_l\right)
    \right|^4 .
\]
This quantity is invariant under the equivalence relation for unitary error bases in Eq.~\eqref{eq:equiv}.
Indeed, if $U_i=e^{\ii\theta_i}A\cdot V_{\pi(i)}\cdot B$, then
\[
    U_i^\dagger U_j U_k^\dagger U_l
    =
    e^{\ii(-\theta_i+\theta_j-\theta_k+\theta_l)}
    B^\dagger
    V_{\pi(i)}^\dagger V_{\pi(j)}
    V_{\pi(k)}^\dagger V_{\pi(l)}
    B .
\]
Taking the trace and then the absolute value removes, respectively, the conjugation by $B$ and the phase factor.
Moreover, the sum over all elements is unaffected by the permutation $\pi$.

Set $q=e^{3\ii\varphi}$.
For $\mathcal U_\varphi$, an exact symbolic enumeration of the $9^4$ traces gives the following table
\[
\begin{array}{c|ccccc}
    |T_{ijkl}|^4 & 0 & 3^4 & |P_0(q)|^4 & |P_1(q)|^4 & |P_2(q)|^4 \\ \hline
    \text{Number of occurrences} & 5184 & 405 & 324 & 324 & 324
\end{array},
\]
where $T_{ijkl}:=\Tr(U_i^\dagger U_j U_k^\dagger U_l)$ and $P_r(q)=1+\omega^r q+\omega^{2r}q^2$.
Therefore
\[
    I(\mathcal U_\varphi)
    =
    405\cdot 3^4
    +
    324\sum_{r=0}^2 |P_r(q)|^4 
    = 32805 + 18438 + 3888(q^3+q^{-3})
    = 51273+7776\cos(9\varphi).
\]
If the bases for $\varphi_1$ and for $\varphi_2$ are locally unitarily equivalent, Lemma~\ref{lem:meb-ueb-equivalence} implies that $\mathcal U_{\varphi_1}$ and $\mathcal U_{\varphi_2}$ are equivalent.
Hence $\cos(9\varphi_1)=\cos(9\varphi_2)$, which for real $\varphi_1$ and $\varphi_2$ is equivalent to $\varphi_1\pm\varphi_2=2\pi k/9$, with integer $k$.

It remains to prove sufficiency.
The diagonal unitary $D=\operatorname{diag}(1,e^{\ii\frac{2\pi}{9}},e^{\ii\frac{4\pi}{9}})$ satisfies, up to global phases,
\[
    D Z^a D^\dagger=Z^a,\qquad
    D X_+^{(\varphi)}D^\dagger\sim X_+^{(\varphi+\frac{2\pi}{9})},\qquad
    D X_-^{(\varphi)}D^\dagger\sim Z X_-^{(\varphi+\frac{2\pi}{9})}.
\]
Therefore, conjugation by $D$ maps $\mathcal U_\varphi$ projectively onto $\mathcal U_{\varphi+\frac{2\pi}{9}}$.
Similarly, the unitary
\[
    R=\begin{pmatrix}
        0&1&0\\
        e^{\ii\frac{2\pi}{9}}&0&0\\
        0&0&e^{-\ii\frac{2\pi}{9}}
    \end{pmatrix}
\]
satisfies, up to global phases,
\[
    R Z^a R^\dagger\sim Z^{-a},\qquad
    R X_+^{(\varphi)}R^\dagger\sim X_-^{(\frac{2\pi}{9}-\varphi)},\qquad
    R X_-^{(\varphi)}R^\dagger\sim Z X_+^{(\frac{2\pi}{9}-\varphi)}.
\]
Thus conjugation by $R$ maps $\mathcal U_\varphi$ projectively onto $\mathcal U_{\frac{2\pi}{9}-\varphi}$.
These two equivalences generate all transformations $\varphi\mapsto\pm\varphi+2\pi n/9$, completing the proof.
\end{proof}

\section{The Clifford hierarchy}\label{App:Clifford}
In this section, we summarize some properties of the Clifford hierarchy for qudits \cite{cui2017}.
The Clifford hierarchy is a recursive algebraic hierarchy that organizes unitaries according to their action on the generalized Pauli group.
Beyond its relevance in quantum computing \cite{gottesman1998}, it has found applications in the context of quantum position verification \cite{Chakraborty2015} and nonlocal computation \cite{Gatto_Lamas_2023} because it is associated to the nonlocal complexity of unitary operations.

The Clifford hierarchy is formulated as a collection of nested subsets $\mathcal{C}^{(1)} \subset \mathcal{C}^{(2)} \subset \mathcal{C}^{(3)} \subset \dots \subset \mathcal U(d^n)$ of unitary operators acting on $n$ qudits \cite{Gottesman_1999}.
Specifically, the hierarchy is defined recursively as follows: Let the first level $ \mathcal{C}_1$ correspond to the generalized Pauli group $\mathcal P$, which is defined according to
\begin{equation}
    \mathcal{C}_1\equiv \mathcal P \equiv \{ D_{i_1j_1}\otimes \dots \otimes  D_{i_n j_n}\},
\end{equation}
where $D_{i_k j_k} = \omega^{i_k j_k} X^{i_k} Z^{j_k}$ is a Weyl-Heisenberg operator and $(i_k,j_k) \in \{ 0,\dots,d-1\}^2$ for all $k = 1,\dots,n$.
Here, $\omega = e^{2\pi i/d}$, and the group generators are given by $X  =\sum_{k=0}^{d-1} \ketbra{k \oplus 1}{k}$ and $Z = \sum_{k=0}^{d-1} \omega^{k} \ketbra{k}{k}$, where $\oplus$ denotes addition modulo $d$.
The second subset $\mathcal{C}_2$, known as the Clifford group, is defined as the set of unitaries that normalize the Pauli group $\mathcal P$, namely
\begin{equation}
\mathcal{C}_2 \equiv \{U \in \mathcal{U}(d^n): U \mathcal{P} U^\dagger \in \mathcal{P}\}.
\end{equation}
For $k > 2$, each level $\mathcal{C}_k$ is defined recursively as follows
\begin{equation}
    \mathcal{C}_k \equiv \{U \in \mathcal{U}(d^n): U \mathcal{P} U^\dagger \in \mathcal{C}_{k-1}\}.
    \label{eq:Ck}
\end{equation}
In general the sets $\mathcal{C}_k$ does not form a group for $k >2$.
However, an important special case is that of diagonal unitaries.
The subset of diagonal unitaries $\mathcal C_k^d \subset\mathcal C_k$ in the $k$-th level of the Clifford hierarchy do form a group \cite{cui2017}.

\subsection{Clifford level of diagonal unitaries}\label{App:CliffordDiagonal}
In general, it is difficult to fully characterize the unitaries that belong to a certain level of the Clifford hierarchy.
However, an important special case is that of diagonal unitaries, which has been completely characterized in the case of qudits of prime dimensions \cite{cui2017}.

Consider that $D$ is a diagonal unitary acting on $n$ qudits of prime dimension $d$.
Then, the unitary can be approximated as \cite{cui2017}
\begin{equation}\label{eq:diag_unitary}
D \ket{\vec{z}} = \exp\!\left(2 \pi \ii \frac{f_m(\vec{z})}{d^m}\right) \ket{\vec{z}},
\end{equation}
where $m \in \mathbb{N}$ denotes the precision, and $\ket{\vec{z}} = \ket{z_1,\dots,z_n}$ is an $n$-qudit computational basis state with each $z_j \in \{0,\dots,d-1\}$.
The function $f_m(\vec{z})$ is an integer-valued polynomial taking values in $\{0,1,\dots,d^m-1\}$, that is,
\[
f_m : \vec{z} \mapsto p, \qquad p \in \{0,1,\dots,d^m-1\}.
\]

We now describe how to construct the so-called \textit{phase polynomial} $f_m$ \cite{cui2017}.
To this end, define the parameter list $L = \{1,z_1,\dots z_n\}$.
Over $L$, construct a basis, $\mathcal{S}$, of monomials of degree at most $n(d-1)$.
Then, any phase polynomial $f_m$ can be expanded as
\begin{equation}
f_m(\vec z)=\sum_{\tilde z\in \mathcal S} a_{\tilde z}\,\tilde z \pmod{d^m},
\end{equation}
where the coefficients $a_{\tilde z}$ belong to $\mathbb Z_{d^m}$.
Since the constant monomial $1$ contributes only an overall global phase to $D$ in Eq.~\eqref{eq:diag_unitary}, it may be omitted.
We therefore restrict to the non-constant monomials, i.e., $\tilde{\mathcal S}=\mathcal S\setminus\{1\}$.

Importantly, it was proven that the lowest level $k$ for which $D \in \mathcal{C}_k$ can be determined through its phase polynomial $f_m(\vec{z})$, as follows \cite{cui2017}
\begin{equation}\label{eq:Clifford_level}
k = \max_{\tilde{z} \in\tilde{\mathcal{S}}} \ \big[(d-1)(m-\nu_d(a_{\tilde z})-1)+|\tilde z| \big].
\end{equation}
Here, the maximization runs over all monomials $\tilde z$ in $\tilde{\mathcal{S}}$, where $|\tilde z|$ denotes its degree, $m$ is the precision, and $a_{\tilde z}$ is the associated coefficient.
Moreover, $\nu_d(a_{\tilde z})$ is the $d$-adic valuation of $a_{\tilde z}$.
This is, $\nu_d(a_{\tilde z})$ is the exponent of the highest power of $d$ dividing $a_{\tilde z}$ such that the quotient is an integer.

\section{Localization of the measurements in the family}\label{App:localization}
We briefly summarize a teleportation-based localization protocol, relevant to bound the minimal entanglement cost needed to locally implement the measurements \cite{Pauwels2025}.
This protocol can be formulated in terms of a complexity hierarchy, the so-called \textit{Vaidman hierarchy}, and leads to conditions which are hard to solve but in principle can be checked recursively.
However, because the measurements that we consider (recall Eqs.~\eqref{eq:family} and \eqref{eq:family_SIC}) have a local tetrahedral structure, the localizability can be analysed in a much simpler way via the Clifford hierarchy.
This alternative approach, however, only provides upper bounds to the amount of entanglement required to localize the measurement through the Vaidman hierarchy.

We now describe how both approaches (checking the position of the measurement in the Vaidman hierarchy, and checking its position in the Clifford hierarchy) can be used to determine the amount of entanglement needed to simulate the measurement statistics $p(j) =  |\langle j |\M^\dagger|\psi_{AB}\rangle|^2$ for any two-qutrit state $\ket{\psi_{AB}}$.
$\M$ is the unitary matrix that rotates the computational basis into the measurement basis (see, e.g., Eq.~\eqref{eq:familyStates} in Appendix~\ref{app:M}).

\subsection{The Vaidman hierarchy with finite entanglement consumption}
To bound the minimal entanglement cost to localize a given measurement, we first consider the teleportation-based localization protocol of Ref.~\cite{Pauwels2025}.
The main idea behind this protocol is to use a finite number of rounds of blind back-and-forth teleportation without communication during the protocol, and a final broadcasting of classical communication to reconstruct the measurement statistics.

Consider two distant parties, Alice and Bob, that share the two-qutrit state $\ket{\psi_{AB}}$ and a finite number of Bell pairs.
Bob starts by teleporting his part of the state $\ket{\psi_{AB}}$ to Alice by applying a Bell state measurement on this part and his part of the first Bell pair.
This results in Alice holding $\ket{\psi_{AB}}$ up to a local unitary that depends on Bob's measurement outcome.
However, since no communication is allowed, Bob cannot convey which is the relevant correction unitary to Alice.
Instead, Alice directly applies the measurement unitary $\M^\dagger$ to the distorted state, and then teleports the resulting state back to Bob, using her shares of two additional Bell pairs.
Bob does not know which local distortions were caused by Alice's teleportation step, so instead he tries to correct his own distortion from the first round.
Bob then teleports the resulting state back to Alice using one out of $d^2$ teleportation channels, in an attempt to encode his measurement outcome from the previous round.
The parties continue the blind back-and-forth teleportation for a finite number of times, and at each step they try to correct their previous distortions and then use the teleportation channel to encode their previous outcomes.
For each step, the number of Bell pairs required to implement the protocol grows exponentially.
At the last step, one party, let us say Alice, measures her state in the computational basis instead of teleporting it back to Bob.
If the measurement statistics, up to joint deterministic post-processing, match the target distribution associated with the original measurement, the measurement can be localized with the assistance of a finite number of Bell pairs.
Notably, since each round of teleportation introduces unitary distortions, this protocol can only be terminated after a finite number of rounds if the effective unitary becomes non-distorting.

Based on the described teleportation-based localization protocol, one can systematically classify the complexity of locally performing a given measurement in terms of the number of back-and-forth teleportation steps, using the so-called \textit{Vaidman hierarchy} \cite{Pauwels2025}.
Similarly to the Clifford hierarchy, the Vaidman hierarchy is a collection of nested subsets $\mathcal{V}_1 \subsetneq \mathcal{V}_2 \subsetneq\dots \subsetneq \mathcal{U}(d^2)$ of unitary operators acting on two qudits.
The level $\mathcal V_k$ with $k\geq 1$ is defined as the set of measurement unitaries $\M$ that satisfy \cite{Pauwels2025}
\begin{equation}
    \mathcal{V}_k\equiv \{ \lvert \M^\dagger\cdot(\mathbb{1} \otimes\mathcal P)\cdot\M \rvert \in \bar{\mathcal{V}}_k\},
    \label{eq:V2}
\end{equation}
 where $\mathcal P$ denotes the generalized Pauli group and the set $\bar{\mathcal{V}}_k$ is defined recursively as follows
\begin{equation}
    \bar{\mathcal{V}}_k \equiv \{\M \in \mathcal{U}(d^2):| \M^\dagger\cdot(\mathcal P \otimes \mathcal{P})\cdot\M |\in \bar{\mathcal{V}}_{k-1}\},
    \label{eq:Vk}
\end{equation}
with $\mathcal{V}_1 = \bar{\mathcal{V}}_1 = \{P\cdot   \tilde{\Phi} \}$ being the set of permutation matrices with arbitrary phases in their entries.

Importantly, the levels of the Clifford hierarchy $\mathcal{C}_k$ form strict subsets of the corresponding $\mathcal{V}_k$, i.e., $\mathcal{C}_k \subsetneq \mathcal{V}_k$.
The reason for this is that the relations of $\mathcal{V}_k$ are strictly weaker than those of $\mathcal{C}_k$.
Indeed, all the elements of $\mathcal{C}_1=\{\mathcal{P}\otimes\mathcal{P}\}$ are permutation matrices where the entries have some fixed phases, but they do not cover all permutation matrices and all possible phases, so $\mathcal{C}_1\subsetneq\mathcal{V}_1$.
For the rest of the levels, Eq.~\eqref{eq:V2} contains only a strict subset of the conditions in Eq.~\eqref{eq:Ck}, and thus $\mathcal{C}_k\subsetneq\mathcal{V}_k$ as well.

To determine the level $k$ in the Vaidman hierarchy a measurement $\M$ belongs to, one first checks whether the measurement satisfies Eq.~\eqref{eq:V2} for $k=2$ (i.e., if all the conjugations of $\mathbb{1}\otimes(\text{Pauli matrix})$ with $\M$ are permutation matrices with phases).
If this is the case, $\mathcal{M}$ belongs to the second level in the Vaidman hierarchy.
If not, one continues and studies whether it belongs to $\mathcal V_3$, by checking whether all conditions in Eq.~\eqref{eq:Vk}.
Again, if all conditions hold, $\M \in \mathcal V_3$.
Otherwise, one continues to check the next level of the hierarchy, repeating this procedure until termination or saturation of the available resources.
Due to the exponential growth of the different combinations of distortions, from a computational point of view it is only efficient to establish whether a measurement belongs to the first few levels in the hierarchy.

We apply the procedure above to the measurement bases in Eq.~\eqref{eq:familyM} for different values of $\varphi$.
For a discrete set of values for $\varphi$, we find that the maximally entangled measurements can be localized at low levels for $\varphi \in \{0,2 \pi/27\}$, namely at $k = 2$ and $k=5$, see Table~\ref{tab:vaidman}.
Moreover, for  $\varphi \in \{2 \pi/81,4 \pi/81,8 \pi/81\}$ we can establish that its position in the Vaidman hierarchy is $k >6$, since we observe that the conditions corresponding to $k=6$ are not satisfied, and going beyond the search program becomes too large to run in reasonable time on an ordinary computer.

\begin{table}[hb!]
    \begin{tabular}{l|llllllll}
        \multicolumn{1}{l|}{Phase, $\varphi$}  & 0 \qquad \quad & $\frac{2 \pi}{81}$ \qquad \quad  & $\frac{4 \pi }{81}$\qquad \quad & $\frac{2 \pi}{27}$\qquad \quad  & $\frac{8 \pi}{81}$  \\[0.2em] \hline
         Vaidman level, $k$                               &    2  &$>6$ & $>6$ & $5$ &  $>6$  \\
    \end{tabular}
    \caption{Vaidman levels for several maximally entangled measurements in the family, using the general membership criterion.
    }
    \label{tab:vaidman}
\end{table}

\subsection{Simplifications for measurements with tetrahedral symmetry}
Given the combinatorial explosion of the previous method, in the following we describe an alternative that, exploiting the fact that the measurements we consider have tetrahedral symmetry, enables us to upper bound the Vaidman level of the measurement.

Recall that the eigenbasis of the maximally entangled measurements in the family (which we will refer to as $\M_\varphi$) can be generated via $\ket*{\Phi^{(\varphi)}_{j}} = U_{j}\otimes U_{j}^* \ket{\Phi_\varphi}$, where the set of generating unitaries $\{U_{j}\otimes U_{j}^*\}_{j}$ corresponds to the abelian group $G = \langle Z \otimes Z^*, X \otimes X\rangle$ of two-qutrit Weyl-Heisenberg operators, and $\ket{\Phi_\varphi}=-\frac{1}{\sqrt{3}}\left(\ket{00} + e^{-\ii \varphi} \ket{12} + e^{\ii \varphi}\ket{21} \right)$.
This implies that the two-qutrit fiducial state $\ket{\Phi_\varphi}$ can be prepared using following circuit \cite{pauwels2025b}:
\begin{equation}
    \ket{\Phi_\varphi} =  S\cdot(\mathbb{1} \otimes F)\cdot D_{\alpha(\varphi)}\cdot F^{\otimes 2} \ket{0}^{\otimes 2}.
    \label{eq:fid_circuit}
\end{equation}
Here, the shift operator $S$, the Fourier matrix $F$, and the diagonal phase gate $D_{\alpha(\varphi)}$ are given by
\begin{equation}
   S = \sum_{j,k=0}^{d-1} \ketbra{j \oplus k,k}{j,k}, \quad  F = \frac{1}{\sqrt{d}} \sum_{j,k = 0}^{d-1} \omega^{jk} \ketbra{j}{k}, \quad D_{\alpha(\varphi)} = \sum_{j,k=0}^{d-1} e^{\ii \alpha_{jk}^{(\varphi)}} \ketbra{j,k}{j,k},
\end{equation}
where $d = 3$, $\oplus$ denotes addition modulo $d$ in general, $\omega = e^{2 \pi i/d}$, and the phases $\alpha_{jk}^{(\varphi)}$ depend on $\varphi$.
Then, the Clifford level of the circuit defined in Eq.~\eqref{eq:fid_circuit} is characterized by the complexity of diagonal phase gate $D_{\alpha(\varphi)}$.
That is, for $k \geq 2$, the Clifford level of the circuit corresponds to the lowest level at which $D_{\alpha(\varphi)} \in \mathcal C_k$.
By extension, this implies that $\M_\varphi \in \mathcal{C}_k$ \cite{pauwels2025b}.
Importantly, since each level in the Clifford hierarchy is a subset of the corresponding level in the Vaidman hierarchy, $\M_\varphi \in \mathcal{C}_k$ implies that $\M_\varphi \in \mathcal{V}_{\leq k}$.

Moreover, using that $D_{\alpha(\varphi)}$ is a diagonal unitary, its Clifford level can be determined from its associated phase polynomial \cite{cui2017}.
Thus, we now study the phase polynomials associated with the maximally entangled measurements $\M_\varphi$.
Simplifying Eq.~\eqref{eq:fid_circuit}, we find that the fiducial state $\ket{\Phi_\varphi}$ can be directly expressed as follows:
\begin{equation}
    \ket{\Phi_\varphi} =  \frac{1}{d} \sum_{x,z=0}^{d-1} e^{\ii \alpha_{xz}^{(\varphi)}} \ket{\Psi_{xz}},
\end{equation}
where $\ket{\Psi_{xz}} = (X^{x}\cdot Z^{z} \otimes \mathbb{1}) \ket{\Omega^+}$ corresponds to the standard Bell states.
By computing the overlaps between the eigenbasis $\ket{\Psi_{xz}}$ and the fiducial state, we obtain an equation which we can solve for the phases $\alpha_{xz}^{(\varphi)}$, namely
\begin{equation}
    \alpha_{xz}^{(\varphi)} = -\ii\left[\ln(d) +\ln( \braket{\Psi_{xz}}{\Phi_\varphi}) \right].
\end{equation}
Next, we study the diagonal phase gate.
By factoring out the phase $\alpha_{00}^{(\varphi)}$ we obtain
\begin{equation}
     D_{\alpha(\varphi)} = \sum_{x,z=0}^{d-1} e^{\ii \alpha_{xz}^{(\varphi)}} \ketbra{x,z}{x,z} = e^{\ii \tilde \alpha_{00}^{(\varphi)}}\sum_{x,z=0}^{d-1}  e^{\ii \tilde\alpha_{xz}^{(\varphi)}} \ketbra{x,z}{x,z}.
\end{equation}
Presenting the effective phases $\tilde\alpha_{xz}^{(\varphi)} = \alpha_{xz}^{(\varphi)} - \alpha_{00}^{(\varphi)}$ in matrix form, $A_\varphi = \sum_{x,z=0}^{d-1} \tilde\alpha_{xz}^{(\varphi)}\ketbra{x}{z}$, we have
\begin{equation}
    A_\varphi = 
    \begin{pmatrix}
        0 & 0 & 0 \\[0.5em]
        \varphi & -\frac{2\pi}{3}+\varphi  & -\frac{4\pi}{3}+\varphi  \\[0.5em]
        -\varphi  & -\frac{4\pi}{3}-\varphi  &  -\frac{8\pi}{3}-\varphi
    \end{pmatrix}.
    \label{eq:phase_d3}
\end{equation}

We now want to find for which phases $\varphi$ it holds that $D_{\alpha(\varphi)}$ can be written in terms of a phase polynomial with finite precision $m$ \cite{cui2017}.
To this end, define the parameter list $L = \{1,x,z\}$.
Over $L$, we construct the basis $\mathcal{S} = \{1,x,z,xz,x^2,z^2,x^2z,xz^2,x^2z^2\}$.
The phase polynomial $f_m(x,z)$, where $m$ is the precision, is then given by a linear combination over $\mathcal S$.
Thus, we seek phases $\varphi$ and associated integer-valued phase polynomials such that
\begin{equation}\label{eq:phase_poly_rel}
    e^{\ii \tilde \alpha_{00}^{(\varphi)}} \sum_{x,z=0}^{d-1}  e^{\ii \tilde\alpha_{xz}^{(\varphi)}} \ketbra{x,z}{x,z} = e^{\ii \tilde \alpha_{00}^{(\varphi)}} \sum_{x,z=0}^{d-1}  e^{\frac{2\pi\ii}{3^m} f_m(x,z)} \ketbra{x,z}{x,z}.
\end{equation}
Note that Eq.~\eqref{eq:phase_d3} directly implies that $\tilde \alpha_{0z}^{(\varphi)} = f_m(0,z) = 0$ for $z = 0,1,2$.
Therefore, it is sufficient to consider the phase polynomial $f_m(x,z)$ over the reduced set $\tilde{\mathcal{S}} = \{x,xz,x^2,x^2z,xz^2,x^2z^2\}$, where
\begin{equation}
    f_m(x,z) = a_0x+ a_1 x^2+a_2 xz+  a_3 x^2z+ a_4xz^2+ a_5x^2z^2 \mod 3^m.
\end{equation}
We now want to find coefficients $a_j \in \mathbb Z_{3^m}$ for $j=0,\dots,5$ such that the following relation holds
\begin{equation}\label{eq:rel_phase}
    f_m(x,z) = \frac{3^m}{2 \pi}\tilde\alpha_{xz}^{(\varphi)} \mod 3^m,
\end{equation}
In matrix representation, the equation we want to solve, modulo $3^m$, reads
\begin{equation}
    \begin{pmatrix}
        0 & 0 & 0 \\[0.5em]
        a_0 & a_2  & a_4  \\[0.5em]
        a_1  & a_3  & a_5
    \end{pmatrix}
    =
    \begin{pmatrix}
        0 & 0 & 0 \\[0.5em]
        \frac{3^m}{2 \pi} \varphi & -3^{m-1}+\frac{3^m}{2 \pi}  \varphi  & -2\times3^{m-1}+\frac{2 \pi}{3^m}\varphi  \\[0.5em]
        -\frac{3^m}{2 \pi}\varphi  & -2\times3^{m-1}-\frac{3^m}{2 \pi}\varphi  &  -4\times3^{m-1}-\frac{3^m}{2 \pi}\varphi
    \end{pmatrix} \mod 3^m.
    \label{eq:phases}
\end{equation}
Clearly, for the coefficients $\{a_j\}$ to be integers in the cyclic group $\mathbb Z_{3^m}$ for finite $m$, we need that 
\begin{equation}\label{eq:angles}
    \varphi = \frac{2 \pi j}{3^l},
\end{equation}
where $j$ and $l$ are non-negative integers.
Note, that for $l = 0, 1,2$, the measurements associated with $\varphi = 2\pi j$, $\varphi = \frac{2 \pi}{3}j$ and $\varphi = \frac{2 \pi}{9}j$ are equivalent to the standard Bell state measurements for all integers $j$.
Inserting Eq.~\eqref{eq:angles} into Eq.~\eqref{eq:phases} gives that
\begin{equation}
    \begin{pmatrix}
        0 & 0 & 0 \\[0.5em]
        a_0 & a_2  & a_4  \\[0.5em]
        a_1  & a_3  & a_5
    \end{pmatrix}
    =
    \begin{pmatrix}
        0 & 0 & 0 \\[0.5em]
        3^{m-l}j & -3^{m-1}+3^{m-l}j  & -2\times3^{m-1}+3^{m-l}j  \\[0.5em]
        -3^{m-l}j  & -2\times3^{m-1}-3^{m-l}j  &  -4\times3^{m-1}-3^{m-l}j
    \end{pmatrix} \mod 3^m.
    \label{eq:final_phases}
\end{equation}
By fixing $j$, $l$, and the precision $m \geq 1$, Eq.~\eqref{eq:final_phases} gives six independent equations from we can solve for the coefficients $a_j$ of $f_m(x,z)$.
Note that, in general, one can restrict $m \geq l$.
The reason for this is that, if $m <l$, one needs to require $j \propto 3^{l-m}$ to ensure that all coefficients are integers.

When a solution to the set of linear equations exists, the Clifford level $\mathcal{C}_k$ of the diagonal gate $D_{\alpha(\varphi)}$ can easily be determined from Eq.~\eqref{eq:Clifford_level} in Appendix~\ref{App:CliffordDiagonal}.
By extension, this gives the Clifford level of the measurement.

In the following, we perform a systematic study of phases associated with diagonal gates at low levels in the Clifford hierarchy.
Due to the equivalences between different measurements in the family under local unitaries (recall that measurements for $\varphi$, $\varphi+\frac{2\pi}{9}$ and $\frac{2\pi}{9}-\varphi$ are equivalent), we restrict the study to $\varphi \in [0,\pi/9]$.

\paragraph{Precision: \textit{m} = 1}
First, let the precision be $m = 1$.
In this case $l \in \{0,1\}$.
For $l=0$ we find that, modulo 3, there is only one equation (and this is independent of $j$), namely
\begin{equation}
    \begin{pmatrix}
        0 & 0 & 0 \\[0.5em]
        a_0 & a_2  & a_4  \\[0.5em]
        a_1  & a_3  & a_5
    \end{pmatrix}
    =
    \begin{pmatrix}
        0 & 0 & 0 \\[0.5em]
        0 & 2  & 1  \\[0.5em]
        0 & 1  &  2
    \end{pmatrix}.
    \label{eq:m1l0}
\end{equation}
Solving this equation we find that the only non-zero coefficient is $a_2 = 2$, which results in the phase polynomial $f_{m = 1}(x,z) = 2xz \mod 3$.
From this, we can compute the associated Clifford level $k$ using Eq.~\eqref{eq:Clifford_level}, finding that $k = 2$.
This is expected, since the measurement associated with this phase, $\varphi = 2 \pi j$, corresponds to the generalized Bell state measurement.

Next, we consider $l = 1$.
In this case, $\varphi = \frac{2 \pi j}{3}$.
This is an integer valued multiple of $\frac{2\pi}{9}$ for all $j$, and thus the associated measurements are all equivalent to the generalized Bell state measurement.

\paragraph{Precision: \textit{m} = 2}
Next, we consider the precision $m = 2$.
Again, if $l = 0,1,2$, the measurement associated with the phase $\varphi = \frac{2 \pi j}{3^{l}}$ is equivalent to the Bell state measurement.
Thus, we do not find a phase $\varphi$ associated with a new measurement.

\paragraph{Precision: \textit{m} = 3}
Next, we consider $m = 3$.
Based on the previous discussion, we may directly consider $ l = 3$, which gives $\varphi = \frac{2 \pi j}{3^3}$.
Note that only $j = 1$ gives new phases within the interval $[0,\pi/9]$.
Thus, the equation we aim to solve reads
\begin{equation}
    \begin{pmatrix}
        0 & 0 & 0 \\[0.5em]
        a_0 & a_2  & a_4  \\[0.5em]
        a_1  & a_3  & a_5
    \end{pmatrix}
    =
    \begin{pmatrix}
        0 & 0 & 0 \\[0.5em]
        1 & 19 & 10  \\[0.5em]
        26 & 8  &  17
    \end{pmatrix} \mod 3^3.
    \label{eq:m3l3}
\end{equation}
This yields the equation $f_{m=3}(x,z) = 16 x + 18 xz + 12x^2 \mod 3^3$, i.e., $a_0 = 16, a_1 = 12, a_2 = 18$, whereas $a_3 = a_4 = a_5 = 0$.
Computing the corresponding Clifford complexity, we find that $k = 5$.

\paragraph{Precision: \textit{m} = 4}
Next, we consider $m = 4$.
To this end, we may directly consider $l = 4$ to find a new phase, yielding $\varphi = \frac{2 \pi j}{3^4}$.
For $j = 1,2,4$ we have new phases within the interval $[0,\pi/9]$.
The equation we need to solve in this case is
\begin{equation}
    \begin{pmatrix}
        0 & 0 & 0 \\[0.5em]
        a_0 & a_2  & a_4  \\[0.5em]
        a_1  & a_3  & a_5
    \end{pmatrix}
    =
    \begin{pmatrix}
        0 & 0 & 0 \\[0.5em]
        j & 54+j & 27+j  \\[0.5em]
        -j & 27-j &  54-j
    \end{pmatrix} \mod 3^4
    \label{eq:m4l4}
\end{equation}
For $j = 1$, we find the phase polynomial $f_{m=4}(x,z) = 43x + 54xz +39x^2 \mod 3^4$ with Clifford level $k = 7$.
For $j =2$ we find that $f_{m=4}(x,z) = 5x + 54xz +78x^2 \mod 3^4$ with $k = 7$.
Lastly, for $j = 4$ we find that $f_{m=4}(x,z) =10x + 54xz +75x^2$ with $k = 7$.

\paragraph{Precision: \textit{m} = 5}
For precision $m = 5$, we find nine new relevant phases for $l = 5$, namely $\varphi = \frac{2 \pi j}{3^5}$ with $j = 1,2,4,5,7,8,10,11,13$.
The equation to solve in this case reads
\begin{equation}
    \begin{pmatrix}
        0 & 0 & 0 \\[0.5em]
        a_0 & a_2  & a_4  \\[0.5em]
        a_1  & a_3  & a_5
    \end{pmatrix}
    =
    \begin{pmatrix}
        0 & 0 & 0 \\[0.5em]
        j & -81+j & 81+j  \\[0.5em]
        -j & 81-j &  162-j
    \end{pmatrix} \mod 3^5
    \label{eq:m5l5}
\end{equation}
We summarize the resulting phase polynomials and their associated Clifford level in Table~\ref{tab:phase}.

\begin{table}[ht!]
    \begin{tabular}{c|c|l|c}
        \multicolumn{1}{c|}{$\varphi$} & Precision, $m$ & Phase polynomial, $f_m(x,z)$ & Clifford level, $k$ \\ \hline
        0                             &  1 &   $ 2xz$              &     2  \\
        $\frac{2 \pi}{243}$            & 5  &  $124x + 162xz +120x^2$       & 9   \\  
        $\frac{4 \pi}{243}$            & 5  &  $5x + 162xz +240x^2$       & 9   \\  
        $\frac{2 \pi}{81}$             & 4  &  $43x + 54xz +39x^2$       & 7   \\  
        $\frac{8 \pi}{243}$            & 5  &  $10x + 162xz +237x^2$       & 9   \\ 
        $\frac{10 \pi}{243}$            & 5  &  $134x + 162xz +114x^2$       & 9   \\ 
        $\frac{4 \pi}{81}$             & 4  &  $5x + 54xz +78x^2$       & 7   \\  
        $\frac{14 \pi}{243}$            & 5  &  $139x + 162xz +111x^2$       & 9   \\ 
        $\frac{16 \pi}{243}$            & 5  &  $20x + 162xz +231x^2$       & 9   \\ 
        $\frac{2 \pi}{27}$             & 3  &  $16x + 18xz +12 x^2$       & 5   \\  
        $\frac{20 \pi}{243}$            & 5  &  $25x + 162xz +228x^2$       & 9   \\ 
        $\frac{22 \pi}{243}$            & 5  &  $149x + 162xz +105x^2$       & 9   \\ 
        $\frac{8 \pi}{81}$             & 4  &  $10x + 54xz +75x^2$       & 7   \\  
        $\frac{26 \pi}{243}$            & 5  &  $154x + 162xz +102x^2$       & 9\\ 
    \end{tabular}
    \caption{List of the measurements and corresponding phase polynomials that have low Clifford level.}
    \label{tab:phase}
\end{table}

\section{Local models for distributions created out of maximally entangled states and measurements}
\label{app:trianglelocal}
Here we aim to show that, when used in a simple network that distributes maximally entangled states, the measurements we define in Eq.~\eqref{eq:familyM} cannot produce network-nonlocal distributions.
For doing so, we prove a more general statement that does not make use of the particular structure of our family of measurements, but only the fact that their eigenstates are maximally entangled.
Thus, this proof is valid also for other maximally entangled measurements.

We consider the triangle network, where three sources distribute bipartite quantum systems among different pairs of the three parties $A$, $B$ and $C$ (each source to a different pair).
Each of the parties measures the corresponding systems received in order to produce their outcome.
The corresponding Born rule for the distributions produced this way is
\begin{align*}
    p(abc) = \Tr\left[\Big(\ketbra{\gamma}_{A_1B_2} \otimes \ketbra{\alpha}_{B_1C_2} \otimes \ketbra{\beta}_{C_1A_2}\Big)\cdot \left(\ketbra{\psi_{A,a}}_{A_1A_2} \otimes \ketbra{\psi_{B,b}}_{B_1B_2} \otimes \ketbra{\psi_{C,c}}_{C_1C_2} \right)\right],
\end{align*}
where $\ket{\alpha}$, $\ket{\beta}$ and $\ket{\gamma}$ are the states distributed by the sources, the subindices denote the Hilbert spaces in which each of the states is defined and, in the case of projective measurements, $\sum_o\ketbra{\psi_{P,o}}=\mathbb{1}$, $\braket{\psi_{P,o}}{\psi_{P,o'}}=\delta_{o,o'}$.
Moreover, we will consider the case of maximally entangled measurements, i.e., of measurements with maximally entangled eigenstates, so we have $\Tr_1\ketbra{\psi_{P,o}}=\Tr_2\ketbra{\psi_{P,o}}=\mathbb{1}/d$ for all parties $P$ and all outcomes $o$.

For all the distributions that admit this form, let us compute the two-body marginal distribution $p(a,b)$.
\begin{align*}
    p(ab) &= \sum_c p(abc) \\
    &= \sum_c \Tr\left[\Big(\ketbra{\gamma}_{A_1B_2} \otimes \ketbra{\alpha}_{B_1C_2} \otimes \ketbra{\beta}_{C_1A_2}\Big)\cdot \left(\ketbra{\psi_{A,a}}_{A_1A_2} \otimes \ketbra{\psi_{B,b}}_{B_1B_2} \otimes \ketbra{\psi_{C,c}}_{C_1C_2} \right)\right],
\end{align*}

We can perform the sum, which gives $\sum_c \ketbra{\psi_{C,c}}_{C_1C_2}=\mathbb{1}_{C_1C_2}$.
Then, splitting the trace in partial traces we obtain
\begin{equation*}
    p(ab) = \Tr_{A_1A_2B_1B_2}\left[\Big(\ketbra{\gamma}_{A_1B_2} \otimes \Tr_{C_2}\ketbra{\alpha}_{B_1C_2} \otimes \Tr_{C_1}\ketbra{\beta}_{C_1A_2}\Big)\cdot \left(\ketbra{\psi_{A,a}}_{A_1A_2} \otimes \ketbra{\psi_{B,b}}_{B_1B_2} \right)\right].
\end{equation*}

We assume that the states $\ket{\alpha}$ and $\ket{\beta}$ are maximally entangled.
Thus, their partial traces are maximally mixed, i.e., $\Tr_{C_2}\ketbra{\alpha}_{B_1C_2}=\frac1d\mathbb{1}_{B_1}$ and $\Tr_{C_1}\ketbra{\beta}_{C_1A_2}=\frac1d\mathbb{1}_{A_2}$.
Inserting these above, we arrive at
\begin{equation*}
    p(ab) = \frac{1}{d^2}\Tr_{A_1B_2}\left[\ketbra{\gamma}_{A_1B_2} \left(\Tr_{A_2}\ketbra{\psi_{A,a}}_{A_1A_2} \otimes \Tr_{B_1}\ketbra{\psi_{B,b}}_{B_1B_2} \right)\right].
\end{equation*}
Since the eigenstates of the measurements are also maximally entangled, we have that $\Tr_{A_2}\ketbra{\psi_{A,a}}_{A_1A_2}=\frac1d\mathbb{1}_{A_1}$ and $\Tr_{B_1}\ketbra{\psi_{B,b}}_{B_1B_2}=\frac1d\mathbb{1}_{B_2}$ as well.
Thus, the two-body probability distribution is
\begin{equation*}
    p(ab) = \frac{1}{d^4}\Tr_{A_1B_2}\left(\ketbra{\gamma}_{A_1B_2}\right)=\frac{1}{d^4}.
\end{equation*}
In other words, any distribution generated in the no-input triangle network when the sources distribute maximally entangled qutrit states in $\mathbb{C}^d\otimes\mathbb{C}^d$ and the parties perform maximally entangled measurements satisfies that the two-body distribution is always uniformly random.
This type of distributions (tripartite with uniform two-body marginals) admits a local model in the triangle network for any cardinality.
In fact, one of the sources is not necessary: it suffices to have the source between parties $A$ and $C$ to produce the outcome $a$, and the source between parties $B$ and $C$ to produce the outcome $b$.
Then, parties $A$ and $B$ just output the respective information they receive, while party $C$ (who receives $a$ and $b$ as well) produces her outcome according to the conditional distribution $p(c|ab)$.

This construction generalizes to networks where at least one party shares states with every other party.
In such case, every distribution in a network for $n$ parties where the $n-1$-partite marginals are uncorrelated (such as the case above of distributions created out of maximally entangled states and maximally entangled measurements) admits a local model in the network.
Note, however, that this construction does not generalize immediately to networks where there are subsets of parties that do not share states (an example being, for instance, the square network), since there is no party that can receive the outputs of all the remaining ones via the hidden variables.
In these cases, it is not known whether non-Clifford maximally entangled measurements such as the family we describe in the main text can produce nonlocality in networks that distribute maximally entangled states.
\end{document}